\documentclass[]{OUP-EJ}
\usepackage{amsmath,amsfonts,amssymb,amsthm,booktabs,color,graphicx,hyperref,url,float}
\allowdisplaybreaks [4]

\begin{document}
\noshorttitletrue
\title[Determine the number of factors]{Determining the number of factors in a large-dimensional generalised factor model}

\author[1]{Rui Wang}
\author[2,*]{Dandan Jiang}

\affil[1]{School of Mathematics and Statistics, Xi'an Jiaotong University, Xi'an, China.}
\affil[2,*]{School of Mathematics and Statistics, Xi'an Jiaotong University, Xi'an, China.}

\corres[*]{Dandan Jiang, School of Mathematics and Statistics, Xi'an Jiaotong University, Xi'an, China. Email: jiangdd@mail.xjtu.edu.cn}

\begin{abstract}
This paper proposes new estimators of the number of factors for a generalised factor model with more relaxed assumptions than the strict factor model.
Under the framework of large cross-sections $N$ and large time dimensions $T$, we first derive the bias-corrected estimator $\hat \sigma^2_*$ of the noise variance in a generalised factor model by random matrix theory. 
Then we construct three information criteria based on $\hat \sigma^2_*$, further propose the consistent estimators of the number of factors. 
Finally, simulations and real data analysis illustrate that 
our proposed estimations  are  more accurate and avoid the  overestimation in some existing works.

\end{abstract}



\maketitle

\section{Introduction}\label{isect1}

With the rapid development of information technology, the factor models with large cross-sections and time-series dimensions 
have emerged more intensively in the fields of economy and finance. For example, it can be used to determine the tick-by-tick transaction prices of a large number of assets,  and to depict the leading eigenvalues of the coincident indexes in macroeconomics, and so on. 
The factor model reveals the  intricate  relationship of the mass of variables through several common factors and simplifies the model structure. \vspace{-2mm}

A critical problem of  factor models is to determine the number  of  static factors or  dynamic factors.  ``Dynamic'' refers to whether the common factor $\mathbf F_t$ itself is modeled as a dynamic process. If it is assumed that $\mathbf F_t$ does not have an auto-correlation structure, i.e. $\mathbb E(\mathbf F_t\mathbf F_s)=0, \forall t\neq s$, the dynamic factor is transformed into a static factor.
There is a lot of literature on this issue for both
 static and dynamic factor models. On one hand, \cite{BN2002},  \cite{On2006, On2010}, \cite{Aless2010}, \cite{AH2013}, \cite{CH2014} studied on the estimation of the number of static factors.
 On the other hand, \cite{FH2000}, \cite{HL2007}, \cite{AW2007}, \cite{BN2007}, \cite{On2009}, etc. investigated on the determination of the the number of dynamic factors. 
 Moreover, the factor  models are closely related to principal component analysis (PCA)  models  and spiked models in random matrix theory.  Many related works to determine the number of factors/principal components/spikes are  developed, such as  \cite{KN2008}, \cite{US2008}, \cite{JL2009}, \cite{PY2012}, \cite{Pass2017}, etc. \vspace{-2mm}

In this paper, we focus on the static factor model and 
propose  new estimators of the number of factors by random matrix theory, as both  the cross-section units $N$ and  time series observations $T$ go to infinity. Within this context,
 a pioneering work is developed by \cite{BN2002}, which provided some information criteria for estimating the number of factors and established the consistency of the estimators of the  number of factors  as $N, T \to \infty$ simultaneously. 
 Following their work, \cite{Aless2010} improved their criteria by introducing  a tuning multiplicative constant in the penalty. More recently, \cite{Pass2017} modified the information criteria to determine the number of factors for the  strict factor model. \vspace{-2mm}
 
 However, these existing works are constrained by different reasons. 
 Practical analysis shows that the information criteria developed by \cite{BN2002} often led to non-robust estimations,  i.e. the number of factors may be overestimated  (see e.g. the application on U.S. macroeconomic data in \cite{Forni2009}). \cite{Aless2010}  considered a factor model with the idiosyncratic components only being  mildly cross-correlated.
  \cite{Pass2017} required the idiosyncratic components $\left\{\mathbf{e}_{t}\right\}_{1 \leqslant t \leqslant T}$ to be independent  and  the population covariance of the observations to be a finite-rank perturbation matrix on the identity matrix.  \vspace{-2mm}
 
The main contributions of our work are reflected in the following points.
First, we relax the  independent distributed assumptions of  $\left\{\mathbf{e}_{t}\right\}_{1 \leqslant t \leqslant T}$ in \cite{Pass2017}, and generalise the strict factor model to a more general form. Thus,  for our target model,
the population covariance matrix of the observations can be regarded as a more generalised spiked population covariance matrix as  mentioned in  \cite{JB2021}. 
Second, we establish the new information criteria by random matrix theory, and propose  more accurate estimators of the number of factors. Compared with the existing  works, our proposed estimators provide the smaller standard errors as illustrated in the simulations. Moreover, we also  prove the consistency of our estimators as $N$ and $T$ approach infinity.
Finally, although our method is constructed under the framework of large $N$ and $T$, the estimations of the number of factors are still robust  even if both $N$ and $T$ are small. As shown in simulation study, when $N=10$ and $T=50$, our estimations  are much closer to the true number of factors, while the estimations in \cite{BN2002} fail for finite samples. \vspace{-2mm}

The arrangement of this article is as follows. First, we generalise the strict  factor model and introduce the  bias-corrected estimator of the noise variance in Section~\ref{isect2}.
Then we construct three new information criteria based on the above  bias-corrected noise estimator, and give the new estimators of  the number of factors for a  generalised factor model in Section~\ref{isect3}. As a by-product, we also prove the consistency of our proposed estimators  as $N, T \to \infty$ simultaneously.
In Section~\ref{isect4}, the Monte Carlo simulations are conducted  to  evaluate the performance  of the proposed  estimators of the number of factors. In Section~\ref{isect5}, we apply our  proposed methods to some  real data sets to illustrate their feasibility  in practice. \vspace{-2mm}

\section{Estimation on the  variance of the noise in  a generalised  factor model}\label{isect2}
\subsection{A bias-corrected  estimation of  the noise variance}\label{isec2.1}
We consider the generalised factor model with the following form
\begin{equation}
	\mathbf X_{t}=\mathbf{\Lambda F}_{t}+ \mathbf{e}_{t} + \boldsymbol{\mu},
	\label{Model1}
\end{equation}
where $\mathbf X_{t}=\left(X_{1t}, X_{2t}, \cdots, X_{Nt}\right)'$ is an $N$-dimensional cross-section vector at time $t$, $\mathbf{\Lambda}$ is an $N \times M$ matrix of factor loading, $\mathbf F_{t}$ is  an $M$-dimensional  vector of common factors, $\boldsymbol{\mu}$ represents the general mean and $\mathbf {e_{t}}$ is an idiosyncratic error vector.  Compared with the strict factor  model, which assumes that the 
covariance matrix of $\mathbf {e_{t}}$ in the model (\ref{Model1}) is $\sigma^{2} \mathbf{I}$,  we extend this assumption to a general case of 
$\sigma^2 \mathbf{V}$, where the matrix $\mathbf{V}$ is a general Hermitian matrix.  The matrix $\mathbf{V}$ satisfies the following points:  First, the eigenvalues of $\mathbf{V}$ are scatted into spaces of several bulks of the general population eigenvalues. Second, the  independent  assumption of $\left\{\mathbf{e}_{t}\right\}_{1 \leqslant t \leqslant T}$ can be removed. 
Thus the model in (\ref{Model1}) is so-called  a generalised factor model. \vspace{-2mm}

To develop meaningful asymptotic theory in the large-dimensional setting,  we assume that both $N$ and $T$ go to infinity proportionally, i.e. $N /T=c_T\rightarrow c>0$, as  $T \rightarrow \infty$.
Therefore, the population covariance matrix of $\left\{\mathbf{X}_{t}\right\}_{1 \leqslant t \leqslant T}$ is \vspace{-2mm}
\begin{align}
	\boldsymbol{\Sigma}=\mathbf{\Lambda} \mathbf{\Lambda}^{\prime}+\sigma^2 \mathbf{V}.
	\label{Sigma}
\end{align}

\vspace{-2mm} To ensure the identification of the model, we impose some assumptions on the model parameters, as mentioned in  \cite{Anderson2003} :
\begin{assumption}
   \vspace{-2mm}
	 $\mathbb{E}\left(\mathbf F_{t}\right)=\mathbf{0}$ and $\mathbb{E}\left(\mathbf F_{t} \mathbf F_{t}^{\prime}\right)=\mathbf{I};$
\end{assumption}
\begin{assumption}
    \vspace{-2mm}
    $\mathbf{\Gamma = \Lambda^{\prime} \Lambda}$ is diagonal matrix of $M$ distinct diagonal eigenvalues.
\end{assumption}
Then the population covariance matrix $\boldsymbol{\Sigma}$ in (\ref{Sigma})  is exactly  the generalised spiked population covariance proposed in   \citep{JB2021}, which has the spectrum form as 
\begin{align}
	\operatorname{spec}(\boldsymbol{\Sigma})
	= \sigma^{2}(\underbrace{\alpha_{1}, \cdots, \alpha_{1}}_{n_{1}}, \cdots, \underbrace{\alpha_{K}, \cdots, \alpha_{K}}_{n_{K}}, \cdots, \underbrace{r_{1}, \cdots, r_{1}, \cdots, r_{s}, \cdots, r_{s}}_{N-M}), \label{spec}
\end{align}
where   $\alpha_{1}, \cdots, \alpha_{K} $ are spikes with multiplicity $n_{k}$, $k = 1,...,K$, respectively, satisfying $n_{1}+\cdots+n_{K}=M$ with the fixed integer $M$. The rest $ r_{1}, \cdots,  r_{s}$  are non-spiked eigenvalues, where $s$ is a fixed small number.  Moreover, we assume that the empirical spectral distribution (ESD) of $\mathbf \Sigma $ converges weakly to a nonrandom probability distribution $H$ on the real line as $N \rightarrow \infty$, which follows a probability distribution and takes the value $r_{i} \sigma^{2}$ in probability $\omega_{i}, i=1, \cdots, s,$ and $\omega_{1}+\cdots+\omega_{s}=1$ . \vspace{-2mm}

The spiked model  describes a phenomenon of a few perturbations to a positively definite matrix, see the references such as \cite{John2001},\cite{BY2008,BY2012}, \cite{JB2021}, etc.
By the close relationship between the factor model and the spiked model, to determine the number the factors is  equivalent to find the number of the spikes in the spiked covariance matrix $\boldsymbol{\Sigma}$. 
Thus we focus on the study of the spiked eigenvalues of the matrix $\boldsymbol{\Sigma}$.
Following the works in \cite{Pass2017}, it is necessary to get an accurate estimation of  $\sigma^{2}$ before estimating the number of factors (or spikes). Then, we refer the work of \cite{Jiang2021},  which  provided a bias-corrected  estimation $\hat{\sigma}_{*}^{2}$ based on random matrix theory. \vspace{-2mm}

Before introducing it, we first provide some preliminary knowledge. Decompose the population covariance matrix as $\boldsymbol{\Sigma}=\mathbf{B}_{N} \mathbf{B}_{N}^{*}$, where $^{*}$ denotes conjugate transposition. Let $\boldsymbol \xi_{t}=\mathbf{B}_{N}^{-1} \left(\mathbf X_{t}-\boldsymbol{\mu}\right)$, then $\mathbf{B}_{N}\boldsymbol{\xi}=\mathbf{B}_{N}(\boldsymbol \xi_{1}, \boldsymbol \xi_{2},\cdots , \boldsymbol \xi_{T})$ can be seen as random samples from the population  covariance matrix $\boldsymbol{\Sigma}$.  And the corresponding sample covariance matrix of  $\mathbf{B}_{N}\boldsymbol{\xi}$  is 
\begin{align}
	\mathbf{S}=\mathbf{B}_{N}\left(\frac{1}{T} \boldsymbol{\xi} \boldsymbol{\xi}^{*}\right) \mathbf{B}_{N}^{*} \label{sample},
\end{align}
which is the  generalised spiked sample covariance matrix. It is should be noted that if the mean parameter $\boldsymbol{\mu}$ is unknown, the sample covariance matrix (\ref{sample}) needs to be replaced by an unbiased form, and the corresponding ratio $c_T = N/T$  should  be replaced by $N/(T-1)$. \vspace{-2mm}

We denote  the set of ranks of  $\sigma^{2}\alpha_{k} $  with multiplicity $n_{k}$ among all the eigenvalues of $\boldsymbol{\Sigma}$ as $J_{k}=\left\{j_{k}+1, \cdots, j_{k}+n_{k}\right\}$, and represent the eigenvalues of  $\mathbf S$ sorted in descending order as $l_{1} \geq l_{2} \geq \cdots \geq l_{N}$. 
According to the classical statistical theory, the maximum likelihood estimator of $\sigma^{2}$ can be obtained as 
\begin{align}
	\begin{split}
		\hat{\sigma}^{2}_{MLE}
		=\frac{1}{(N-M)\left(\omega_{1} r_{1}+\cdots+\omega_{s} r_{s}\right)}\left(\sum_{j=1}^{N} l_{j}-\sum_{j \in \mathcal{J}_{k}, k=1}^{K} l_{j}\right), \label{MLE} \\
	\end{split}
\end{align}
which can be viewed as an appropriate estimation of the noise variance $\sigma^{2}$. However, it is well known that the sample  eigenvalues do not converge to the population ones when the cross-section dimension $N$ is large compared to the time dimension $T$. Therefore,  the estimator (\ref{MLE}) will underestimate the true noise variance $\sigma^2$. \vspace{-2mm}

To this end,  \cite{Jiang2021} established the central limit theorem of $\hat{\sigma}^{2}_{MLE}$ in the large-dimensional setting, and gave the corresponding  bias-corrected  estimation. To refer this work, we define 
$\beta=E\left|\xi_{11}\right|^{4}-q-2$ with $q=1$ for real case and 0 for complex, where $\xi_{11}$ is the first element of $\boldsymbol \xi_{1}=(\xi_{11},\xi_{12}, \cdots,\xi_{1T})$. We denote $F^{c, H}$ as  the LSD of the sample matrix $\mathbf{S}$, and further $\underline{m}(z) \equiv m_{\underline{F}^{c}, H}(z)$  as the Stieltjes Transform of $\underline{F}^{c, H} \equiv(1-c) I_{[0, \infty)}+$
$c F^{c, H} .$	
Then the  bias-corrected  estimator of  the noise variance $\sigma^{2}$ is given in the following proposition. 
\vspace{-4mm}
\begin{proposition}[]
\label{prop1}
	For the factor model (\ref{Model1}) with the spectrum in (\ref{spec}), the bias-corrected estimator $\hat{\sigma}_{*}^{2}$  is given as
	\vspace{-4mm}
	\begin{align}
		\hat{\sigma}_{*}^{2}=\hat{\sigma}^{2}_{MLE}+\frac{\left.b( \alpha_{k}, \hat{\sigma}^{2}_{MLE}\right)-\mu_{x}}{(N-M) \sum_{i=1}^{s} \omega_{i} r_{i}}
		\label{bias},
	\end{align}
	\vspace{-4mm}
	where
	\vspace{-4mm}
    \begin{align}
	b\left(\alpha_{k}, \sigma^{2}\right)=\sum_{k=1}^{K} \sum_{i=1}^{s} \frac{n_{k} c \alpha_{k} \sigma^{2} r_{i} \omega_{i}}{\alpha_{k}-r_{i}} \notag
   \end{align}
   \vspace{-4mm}
	and
	\vspace{-4mm}
	\begin{align*}
		\mu_{x}=&-\frac{q}{2 \pi i} \oint \frac{c  \underline{m}^{2}(z)\left[c \underline{m}(z) \int t\{1+t \underline{m}(z)\}^{-1} d H(t)-1\right] \int t^{2}\{1+t \underline{m}(z)\}^{-3} d H(t)}{\left[1-c \int \underline{m^{2}}(z) t^{2}\{1+t \underline{m}(z)\}^{-2} d H(t)\right]^{2}} d z \\
		&-\frac{\beta c}{2 \pi i} \oint \underline{m}^{2}(z)\left[-1+c \underline{m}(z) \int t\{1+t \underline{m}(z)\}^{-1} d H(t)\right] \\
		& \qquad\cdot \frac{\int t\{1+t \underline{m}(z)\}^{-1} d H(t) \int\{1+t \underline{m}(z)\}^{-2} d H(t)}{1-c \int \underline{m}^{2}(z) t^{2}\{1+t \underline{m}(z)\}^{-2} dH(t)} dz
	\end{align*}
	
\end{proposition}

\subsection{Monte Carlo experiments} \label{isec2.2}
Since \cite{Jiang2021} only performed the simulations for the case of equal non-spikes, we design the following simulation to verify the feasibility of Proposition \ref{prop1} in our model for more complex cases. 
We first set up the following models:
\begin{description}
	\item[Model~1.] Assuming that $\boldsymbol{\Sigma}=\sigma^{2} \mathbf{D}$, where $\sigma^{2}=4$,  $\mathbf{D}=\text{diag}(25, 16, 16, 9, 2,\cdots,2, 1,\cdots,1)$ is an $N \times N$ matrix with  spikes $(25, 16, 16, 9)$ of the multiplicity (1, 2, 1) and  non-spikes 2 of $(N-4)/2$ time, non-spikes 1 of $N/2$ time.
	\item[Model~2.] Assuming that $\boldsymbol{\Sigma}=\sigma^{2} \mathbf{U}\mathbf{D}\mathbf {U^*}$, where 
	$\mathbf{D}$ is defined in Model 1,  $\mathbf{U}$ is composed of eigenvectors of an $N \times N$ matrix $\mathbf{H}\mathbf{H}'$ with the entries of $\mathbf{H}$ being independently sampled from standard Gaussian population. 
\end{description}

 Moreover, for each model,  the Gaussian and Gamma populations are studied to show the conclusion is extensively utilisable without the limitations of population.
\begin{description}
	\item[Gaussian Assumption.]$\left\{\xi_{it}\right\}$ are $i.i.d.$ samples from standard Gaussian population;
	\item[Gamma Assumption.] $\left\{\xi_{it}\right\}$ are $i.i.d.$  samples from $ \left\{ Gamma(2,1)-2\right\} / \sqrt2$ .
\end{description}

Next we will compare $\hat{\sigma}_{*}^{2}$  with $\hat{\sigma}_{MLE}^{2}$ and several other existing noise variance estimators. The definitions of these estimators are given below.

\begin{itemize}
	\item[(a)]
	The estimator $\hat{\sigma}_{\mathrm{P}}^{2}$ in \cite{Pass2017}  is also a bias correction of  the maximum likelihood estimation of the noise,
	but for the case where the non-spikes are all 1 and $\left\{\mathbf{e}_{t}\right\}_{1 \leqslant t \leqslant T}$ are independent, which is defined as 
	$$
	\hat{\sigma}_{\mathrm{P}}^{2}=\tilde{\sigma}^{2}+\frac{c \tilde{\sigma}^{2}}{N-M}\left(M+\sum_{k=1}^{K} \frac{n_{k}}{\alpha_{k}-1}\right),
	$$
	where $\tilde{\sigma}^{2}$ is the maximum likelihood estimation of the noise given in their work.
	\vspace{-2mm}
	\item[(b)]
	The estimator  $\hat{\sigma}_{\mathrm{KN}}^{2}$ in \cite{KN2008} is described  as the solution of the following system of nonlinear equations with $m+1$ unknowns,
    $$
	\hat{\sigma}_{\mathrm{KN}}^{2}-\frac{1}{N-m}\left\{\sum_{j=m+1}^{N} l_{j}+\sum_{j=1}^{m}\left(l_{ j}-\hat{\rho}_{j}\right)\right\}=0
	$$
	and
	$$
	\hat{\rho}_{j}^{2}-\hat{\rho}_{j}\left(l_{j}+\hat{\sigma}_{\mathrm{KN}}^{2}-\hat{\sigma}_{\mathrm{KN}}^{2} \frac{N-m}{T}\right)+l_{ j} \hat{\sigma}_{\mathrm{KN}}^{2}=0, \quad j=1, \cdots, m
	$$
	\vspace{-2mm}
	\item[(c)]
	The estimator $\hat{\sigma}_{\mathrm{US}}^{2}$ in \cite{US2008} is defined as the ratio  of the median of the non-spike sample eigenvalues to the the median of the Marčenko-Pastur distribution $F_{\alpha, 1}$, 
	$$
	\hat{\sigma}_{\mathrm{US}}^{2}=\frac{\operatorname{median}\left(l_{M+1}, \cdots, l_{ N}\right)}{m_{N/T, 1}},
	$$
	where $m_{\alpha, 1}$ is the median of $F_{\alpha, 1}$. 
	\vspace{-2mm}
	\item[(d)]
	The estimator $\hat{\sigma}_{\text {median }}^{2}$ in \cite{JL2009} is defined as the median of the variances across all dimensions of the $T$ samples,
	$$
	\hat{\sigma}_{\text {median }}^{2}=\operatorname{median}\left(\frac{1}{T} \sum_{t=1}^{T} \tilde{X}_{i t}^{2}, \quad 1 \leqslant i \leqslant N\right),
	$$
	where  $\{\tilde{X}_{i t}\}$ are the  centralised data of the original samples $\{{X}_{i t}\}$.
	\vspace{-2mm}
\end{itemize}

\begin{figure}[!t]
	\centering
	\includegraphics[height = 9cm, width = 12cm]{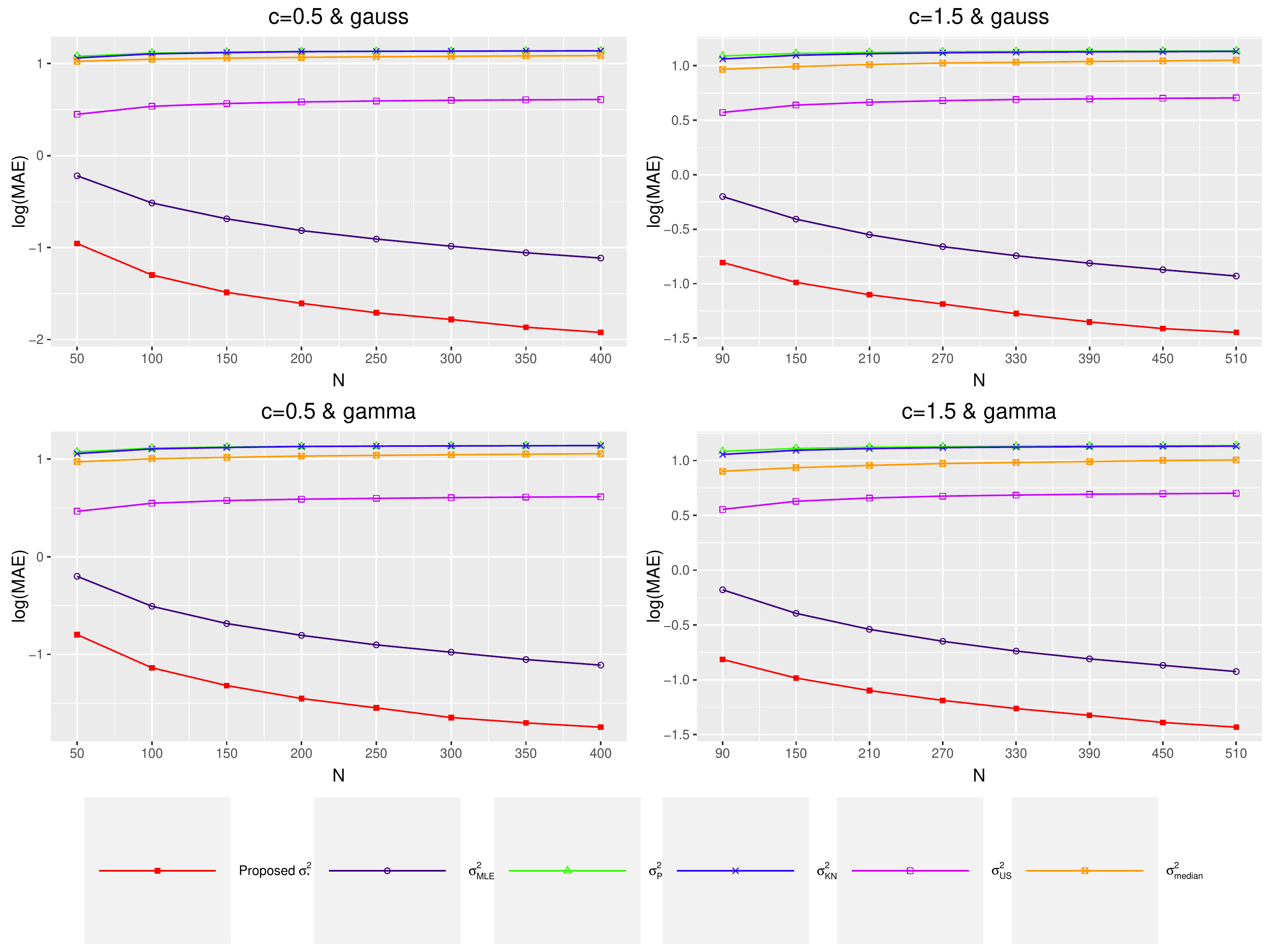}
	\caption{The logarithm transformed MAEs in Model 1  }
	\label{fig-m1}
\end{figure}
\begin{figure}[!t]
	\centering
	\includegraphics[height = 9cm, width = 12cm]{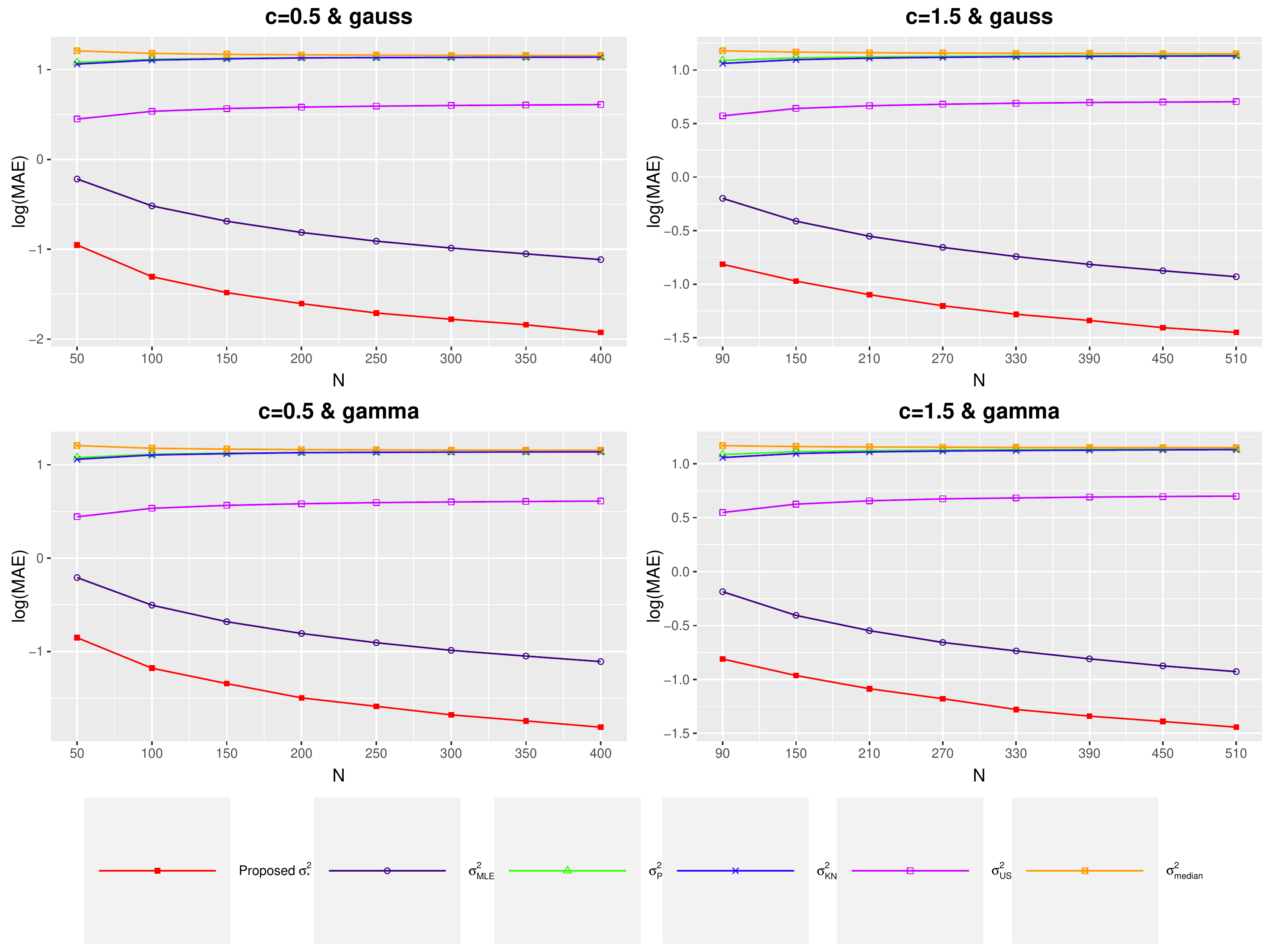}
	\caption{The logarithm transformed MAEs in Model 2  }
	\label{fig-m2}
\end{figure}
We respectively simulate  the  numerical logarithm transformed  mean absolute error (MAE) with 1,000 replications as $N$ increases  for $c=0.5$ and $c=1.5$. When the ratio $c$ is set to be $0.5$, the value of $N$ is set to $50, 100, 150, 200, 250, 300, 350, 400$, respectively. And when the ratio is set to $1.5$, the value of $N$ is set to $90, 150, 210, 270, 330, 390, 450, 510$, respectively. 
 
 The results are reported in Figure \ref{fig-m1} and \ref{fig-m2} . Under different models and distributions, the proposed bias-corrected estimator $\hat{\sigma}_{*}^{2}$ has the smallest logarithm transformed  MAE. As the $N$ increases, the logarithm transformed  MAE of  $\hat{\sigma}_{*}^{2}$ becomes smaller and smaller. However, the estimator  $\hat{\sigma}_{\mathrm{KN}}^{2}$, $\hat{\sigma}_{\mathrm{US}}^{2}$ and  $\hat{\sigma}_{\text {median }}^{2}$ obviously fail to estimate $\sigma^2$ in our model. In addition, the numerical results of $\hat{\sigma}_{\mathrm{P}}^{2}$ and $\hat{\sigma}_{\mathrm{KN}}^{2}$ are  so close that their corresponding curves are not easy to distinguish on the Figure. To show the results more accurately, we put the original simulation values in the Appendix.
 
\section{Estimation on  the number of factors in a  generalised large-dimensional factor model}\label{isect3}
In this section, we will construct the new information criteria based on $\hat{\sigma}^2_*$, and propose the  estimators of  the number of factors for our  generalised factor model.
Recall the work in \cite{BN2002}, 
the common factor $\mathbf{F}$  and the factor loading $\mathbf{\Lambda}$ can be estimated  by the  asymptotic  principal components method in a large panel. The asymptotic principal component method minimises
\begin{align}
	V(m)=\min _{\mathbf \Lambda^m, \mathbf F^{m}}(N T)^{-1} \sum_{i=1}^{N} \sum_{t=1}^{T}\left(X_{i t}-\mathbf \Lambda_{i}^{m^{\prime}} \mathbf F_{t}^{m}\right)^{2}
	\label{V}
\end{align}
subject to the normalisation  $\mathbf F^{m^{\prime}} \mathbf F^{m}/T=\mathbf I_{m}$, where $\mathbf \Lambda^{m}$ is the $N \times m$ factor loading matrix,  $\mathbf F^{m}$ is  $T \times m$  factor matrix, respectively. To be specific, under the normalisation of  $\mathbf F^{m^{\prime}} \mathbf F^{m} / T=\mathbf I_{m}$, we adopt $ \widetilde{\mathbf F}^{m}=\sqrt{T}(\mathbf U_1,\mathbf U_1, \cdots, \mathbf U_m)$ as the
 estimated factor matrix  minimising $V(m)$ , where $\mathbf U_i$ is the eigenvector corresponding to the $i$th largest  eigenvalue of $\mathbf{XX^{\prime}} $. Then, applying  the least square method, we can obtain  the corresponding factor loading matrix, $\widetilde{\mathbf \Lambda}^{m^{\prime}}=(\widetilde{\mathbf F}^{m^{\prime}} \widetilde{\mathbf F}^{m})^{-1} \widetilde{\mathbf F}^{m^{\prime}} \mathbf X=\widetilde{\mathbf F}^{m^{\prime}} \mathbf X / T$.

According to the knowledge of the linear model, the formula (\ref{V}) is a decreasing function of $m$. With the increasing integer $m$, we will get the smaller squared error loss $V(m)$.
But selecting the excessive number $m$ of factors will lose the efficiency of the model, and the simplicity of the model cannot be guaranteed.  For this reason, 
\cite{BN2002} developed  the penalty function $g(N,T)$ relied on both $N$ and $T$ , and avoided over fitting.
Let $V\left(m, \widetilde{\mathbf{F}}^{m}\right)=(NT)^{-1}\min _{\mathbf \Lambda}  \sum_{i=1}^{N} \sum_{t=1}^{T}\{X_{i t}- (\mathbf \Lambda_{i}^{m})' \widetilde{\mathbf F}_{t}^{m}\}^2$, then a loss function $V\left(m, \widetilde{\mathbf F}^m\right)+m g(N, T)$  can be used to determine the number of factors. In order to balance the goodness of fit and simplicity of the model, \cite{BN2002} generalised the $C_{p}$ criterion of Mallows (1973)  and suggested three $PC_{p}$ criteria  under the framework of large  $N$ and  $T$ as follows: 
\begin{align}
	P C_{pj}(m)&=V\left(m, \widetilde{\mathbf F}^{m}\right)+m \hat{\sigma}^{2} g_{j}(N,T),\quad j\in \{1,2,3\},\ \label{PC}
\end{align}	
where  $V\left(m, \widetilde{\mathbf F}^{m}\right)=(NT)^{-1} \sum_{i=1}^{N} \sum_{t=1}^{T}\left(X_{i t}-\widetilde{\mathbf \Lambda}_{i}^{m^{\prime}} \widetilde{\mathbf F}_{t}^{m}\right)^{2}$,  $\hat{\sigma}^{2}$ is a consistent estimator of $ (NT)^{-1} \sum_{i=1}^{N}\sum_{t=1}^{T} \mathbb E(e_{it}^{2})$, and $ g_{j}(N,T)$'s are different penalty functions 
\begin{align*}
	g_{1}(N,T)&=\left(\frac{N+T}{N T}\right) \ln \left(\frac{N+T}{N T}\right) \notag\\
	g_{2}(N,T)&=\left(\frac{N+T}{N T}\right) \ln C_{N T}^{2} \quad\quad\quad, \\
	g_{3}(N,T)&=\left(\frac{\ln C_{N T}^{2}}{C_{N T}^{2}}\right) \notag
\end{align*}
with $C_{N T}=\text{min}\{\sqrt{N},\sqrt{T}\}$.
All of these  penalty functions  satisfy two conditions: (i) $ g(N, T)\rightarrow 0$, (ii) $C^{2}_{NT}g(N, T)\rightarrow \infty$ as $N ,T \rightarrow \infty$ .\\
\indent The $\hat{\sigma}^{2}$ in (\ref{PC}) plays a role as an appropriate scaling parameter for  the penalty term. \cite{BN2002} recommended replacing it with $V\left(m_{0}, \widetilde{\mathbf F}^{m_{0}}\right)$, where $m_{0} $ is a bounded integer such that $m \leqslant m_{0}$.  The estimators of the number of factors   corresponding to the three information criteria are  $\hat{m}_{j}=\arg \min _{0 \leqslant m \leqslant m_{0} } PC_{pj}(m), j\in\{1,2,3\}$ .

As mentioned in the Introduction, the method proposed by \cite{BN2002} often overestimate the number of factors. To improve this problem, we construct the  new information criteria based on the bias-corrected noise  estimator $\hat{\sigma}_{*}^{2}$ in proposition \ref{prop1}. In fact, $V\left(m, \widetilde{\mathbf F}^{m}\right)$ and $\hat{\sigma}^{2}$ are indeed estimations of the noise variance in model (\ref{Model1}). Then we substitute $V\left(m, \widetilde{\mathbf F}^{m}\right)$ in (\ref{PC}) by $\hat{\sigma}_{*}^{2}(m)$, where $\hat{\sigma}_{*}^{2}(m)$ is the function of $m$ with definition in Proposition \ref{prop1}. Moreover, we substitute $\hat{\sigma}^{2}$ by $\hat{\sigma}_{*}^{2}(m_{0})$. Thus, our proposed new information criteria  and the estimators of the number are obtained in the following theorem.

\begin{theorem}[]
	For the determination of the number of factors in the generalised factor model (\ref{Model1}), we propose three information criteria as follows
	\begin{align}
		P C^{*}_{p 1}(m)&=\hat{\sigma}_{*}^{2}(m)+m \hat{\sigma}_{*}^{2}\left(m_{0}\right) \left(\frac{N+T}{N T}\right) \ln \left(\frac{N+T}{N T}\right)\notag\\
		P C^{*}_{p 2}(m)&=\hat{\sigma}_{*}^{2}(m)+m  \hat{\sigma}_{*}^{2}\left(m_{0}\right)\left(\frac{N+T}{N T}\right) \ln C_{N T}^{2} \\
		P C^{*}_{p 3}(m)&=\hat{\sigma}_{*}^{2}(m)+m \hat{\sigma}_{*}^{2}\left(m_{0}\right)\left(\frac{\ln C_{N T}^{2}}{C_{N T}^{2}}\right) \notag
	\end{align}	
	and the corresponding estimators of the number of factors are 
	$$
	\hat{m}_{j}^{*}=\arg \min _{0 \leqslant m \leqslant m_{0}} \mathrm{PC}_{pj}^{*}(m), \quad j \in\{1,2,3\},
	$$
\end{theorem}

Furthermore, we establish the consistency of the corresponding estimators of $M$ as $N ,T\rightarrow \infty$ and give the proof as follows.
\begin{theorem}[]
	Let $0 \leqslant m \leqslant m_{0}$ and $\hat{m}_{j}^{*}=\arg \min _{0 \leqslant m \leqslant m_{0}} \mathrm{PC}_{pj}^{*}(m), j \in\{1,2,3\}$.  Then we have $\lim _{N, T \rightarrow \infty} \operatorname{Prob}[\hat{m}_{j}^{*}=M]=1$, where $M$ is the true number of factors .
	\label{Th2}
\end{theorem}
\begin{proof}
We are going to prove that $\lim _{N, T \rightarrow \infty} P\left\{PC_p^{*}(m)<PC_p^{*}(M)\right\}=0$ for all $m \neq M$ and $m \leqslant m_{0}$, where 
 $PC_p^{*}$ stands for all $PC^{*}_{pj}, j=1,2,3$, because they have the same limiting properties. 
 Since
$$
P C_p^*(m)-P C_p^*(M)=\hat{\sigma}_{*}^{2}(m)-\hat{\sigma}_{*}^{2}(M)-(M-m) \hat{\sigma}_{*}^{2}\left(m_{0}\right) g(N, T) ,
$$
it is sufficient to prove
\begin{align}
   \vspace{-4mm}
	P\left[ \hat{\sigma}_{*}^{2}(M)-\hat{\sigma}_{*}^{2}(m)>(m-M) \hat{\sigma}_{*}^{2}\left(m_{0}\right) g(N, T)\right] \rightarrow 0 \label{ineq1}
\end{align}
   \vspace{-4mm}
or 
   \vspace{-4mm}
\begin{align}
\vspace{-4mm}
	P\left[ \hat{\sigma}_{*}^{2}(m)-\hat{\sigma}_{*}^{2}(M)<(M-m) \hat{\sigma}_{*}^{2}\left(m_{0}\right) g(N, T)\right] \rightarrow 0, \label{ineq2}
\end{align}
   \vspace{-4mm}
as $N, T \rightarrow \infty .$

 Consider first $m>M$. By expression (\ref{bias}), we have
\begin{align*}
\hat{\sigma}_{*}^{2}(M)-\hat{\sigma}_{*}^{2}(m)=\left\{\hat{\sigma}_{MLE}^{2}(M)-\hat{\sigma}_{MLE}^{2}(m)\right\}\left\{1+o_{p}(1)\right\}
\end{align*}
   \vspace{-4mm}
 Moreover,
   \vspace{-4mm}
\begin{align}
	&(N-m)\sum_{i=1}^{s}w_{i}r_{i}\left\{\hat{\sigma}_{MLE}^{2}(M)-\hat{\sigma}_{MLE}^{2}(m)\right\}  \notag\\
	=& (N-M+M-m)\hat{\sigma}_{MLE}^{2}(M)\sum_{i=1}^{s}w_{i}r_{i} -(N-m)\hat{\sigma}_{MLE}^{2}(m)\sum_{i=1}^{s}w_{i}r_{i} \notag \\
	=&\sum_{j=M+1}^{N} l_{j} -\sum_{j=m+1}^{N} l_{j} -(m-M)\sum_{i=1}^{s}w_{i}r_{i} \hat{\sigma}_{MLE}^{2}(M) \notag \\
	=&\sum_{M<i \leqslant m} l_{i}-(m-M)\hat{\sigma}_{MLE}^{2}(M)\sum_{i=1}^{s}w_{i}r_{i} \notag \\
	\leqslant&(m-M)\{l_{M+1}-\hat{\sigma}_{MLE}^{2}(M)\sum_{i=1}^{s}w_{i}r_{i}\} 
	\label{diff1}
\end{align}

Since both $l_{\bar{M+1}}$ and $\hat{\sigma}_{MLE}^{2}(M)$   are bounded positive values, the right-hand side of (\ref{diff1}) is   bounded.
The inequality (\ref{ineq1}) will hold if the penalty satisfies
\begin{align}
	(N-m) g(N,T) >
	\frac{l_{M+1}-\hat{\sigma}_{MLE}^{2}(M)\sum_{i=1}^{s}w_{i}r_{i}}{\hat{\sigma}_{*}^{2}\left(m_{0}\right)\sum_{i=1}^{s}w_{i}r_{i}}
	\label{ineq3}
\end{align}
for large $N$ and $T$.  And we have $C^{2}_{NT}g(N, T)\rightarrow \infty$, the right part of inequality expression (\ref{ineq3}) is a bounded value, and the conclusion follows.\\
\indent Next, for  $m<M$, we have 
\begin{align}
	\hat{\sigma}_{*}^{2}(m)-\hat{\sigma}_{*}^{2}(M)
	=[\hat{\sigma}_{*}^{2}(m)-\hat\sigma_{MLE}^{2}(m)]+[\hat \sigma_{MLE}^{2}(m)-\hat\sigma_{MLE}^{2}(M)] +[\hat \sigma_{MLE}^{2}(M)-\hat{\sigma}_{*}^{2}(M)], \label{split}
\end{align}
where $\hat \sigma_{MLE}^{2}(k)$ represents the maximum likelihood estimation of $\sigma^2$ when the population covariance matrix has $k$ spikes. \\
\indent The first and the third terms on the right side of expression (\ref{split}) are both $O_{p}\left(C_{N T}^{-2}\right)$.   Next, consider the second term. When the  number of real population spikes is $m$ and $M$, there is only a difference of $(M-m)$ spikes between the two populations. We retain  Assumption B about the  factor loadings in \cite{BN2002}, that the factor loadings  grow to $\infty$ with the dimension $N$. 
It implies that the second term is a positive bounded value. Since $g(N, T) \rightarrow 0$ as $N, T \rightarrow \infty $, the inequality (\ref{ineq2})  holds.\\
\indent The conclusion follows.
\end{proof}

\section{Simulation study}\label{isect4}
To check the improvement performance  of our proposed information criteria, Monte Carlo simulations are conducted. we refer the data generating process in \cite{BN2002}, which is expressed as
\begin{align*}
X_{i t}=\sum_{j=1}^{M} \lambda_{i j} F_{j t}+\sqrt{\theta} e_{i t} 
\end{align*}
with the factors $F_{j t}$ being $\mathcal{N}(0,1)$  variates and $\theta=3$.
Different from simulated design in \cite{BN2002} and \cite{Pass2017}, we generalised the settings as $(\lambda_{1j},\lambda_{1j},\cdots,\lambda_{Nj}) \sim \mathcal{N}(\boldsymbol{0}_{N},\mathbf{A})$ and $\{e_{it}\} =\mathbf{V}^{\frac{1}{2}} \{\xi_{it}\}$, where $\mathbf{A}=\text{diag} \left(5,4,4,3,0,\cdots,0\right)$, $\mathbf{V}$ is diagonal or off-diagonal matrix listed in Model 3 to 5, and  $\{{\xi}_{it}\}$  are $i.i.d.$ random variables such that $E\xi_{11}=0,  E|\xi_{11}|^2=1$.


\begin{description}
	\item[Model~3.] Assuming that $\mathbf{V}=\mathbf{D} $, where $\mathbf{D}=\mathbf{I_{N}}$ is an $N \times N$ identity matrix.
	\item[Model~4.]  Assuming that $\mathbf{V}=\mathbf{D}$, and $\mathbf{D}=\text{diag} \left(2,2,\cdots,2 ,1,1,\cdots, 1 \right)$, where eigenvalues 2 and 1 are half and half.
	\item[Model~5.] Assuming that $\mathbf{V} =\mathbf{U}\mathbf{D} \mathbf {U^{*}}$, where $\mathbf{D} $ is defined in Model 4, and $\mathbf{U}$ is defined in Model 2.
\end{description}

By the generalised settings, we relax  the independent or mild cross-correlated  assumptions of
the error sequence $\{\mathbf e_{t}\}_{1\leqslant t \leqslant T }$ than previous works. Furthermore, we reuse the  two population assumptions of $\left\{{\xi}_{it}\right\}$   in subsection \ref{isec2.2}

\begin{table}[!t]
	\centering\caption{Comparison between  $PC_{pj}$ and $PC^{*}_{pj}$ for Model 3 }
	\label{tab:model3}
	\scalebox{0.9}{
	\begin{tabular}{lllllll}
		\hline
	    \multicolumn{1}{c}{($N, T$)} &  $PC_{p1}$   & $PC_{p1}$ & $PC_{p1}$   & $PC^{*}_{p1}$& $PC^{*}_{p2}$ & $PC^{*}_{p3}$\\
		\hline
		\multicolumn{7}{c}{Mode3 under Gaussian assumption } \\
		\multicolumn{1}{l}{$N=100,T=40$} & 4.00  & 4.00  & 4.32(0.48) & 4.00  & 4.00  & 4.00  \\
		\multicolumn{1}{l}{$N=100,T=60$} & 4.00  & 4.00  & 4.16(0.37) & 4.00  & 4.00  & 4.00  \\
		\multicolumn{1}{l}{$N=200,T=60$} & 4.00  & 4.00  & 4.00  & 4.00  & 4.00  & 4.00  \\
		\multicolumn{1}{l}{$N=500,T=60$} & 4.00  & 4.00  & 4.00  & 4.00  & 4.00  & 4.00  \\
		\multicolumn{1}{l}{$N=1000,T=60$} & 4.00  & 4.00  & 4.00  & 4.00  & 4.00  & 4.00  \\
		\multicolumn{1}{l}{$N=2000,T=60$} & 4.00  & 4.00  & 4.00  & 4.00  & 4.00  & 4.00  \\
		\multicolumn{1}{l}{$N=100,T=100$} & 4.00  & 4.00  & 4.72(0.57) & 4.00  & 4.00  & 4.00  \\
		\multicolumn{1}{l}{$N=40,T=100$} & 4.00  & 4.00  & 4.27(0.46) & 4.00  & 4.00  & 4.00  \\
		\multicolumn{1}{l}{$N=60,T=100$} & 4.00  & 4.00  & 4.16(0.36) & 4.00  & 4.00  & 4.00  \\
		\multicolumn{1}{l}{$N=60,T=200$} & 4.00  & 4.00  & 4.00  & 4.00  & 4.00  & 4.00  \\
		\multicolumn{1}{l}{$N=10,T=50$} & 8.00  & 8.00  & 8.00  & 4.01(0.60) & 3.98(0.61) & 4.05(0.60) \\
		\multicolumn{1}{l}{$N=10,T=100$} & 8.00  & 8.00  & 8.00  & 3.92(0.36) & 3.91(0.36) & 3.93(0.33) \\
		\multicolumn{1}{l}{$N=20,T=100$} & 5.50(0.66) & 4.92(0.64) & 6.78(0.66) & 4.00  & 4.00  & 4.00  \\
		\multicolumn{1}{l}{$N=100,T=20$} & 5.53(0.65) & 4.95(0.61) & 6.79(0.65) & 4.32(0.50) & 4.11(0.32) & 5.32(0.77) \\
		\multicolumn{7}{c}{Model3 under Gamma assumption} \\
		\multicolumn{1}{l}{$N=100,T=40$} & 4.00(0.03) & 4.00  & 4.69(0.57) & 4.00  & 4.00  & 4.01(0.11) \\
		\multicolumn{1}{l}{$N=100,T=60$} & 4.00  & 4.00  & 4.47(0.54) & 4.00  & 4.00  & 4.00  \\
		\multicolumn{1}{l}{$N=200,T=60$} & 4.00  & 4.00  & 4.00  & 4.00  & 4.00  & 4.00  \\
		\multicolumn{1}{l}{$N=500,T=60$} & 4.00  & 4.00  & 4.00  & 4.00  & 4.00  & 4.00  \\
		\multicolumn{1}{l}{$N=1000,T=60$} & 4.00  & 4.00  & 4.00  & 4.00  & 4.00  & 4.00  \\
		\multicolumn{1}{l}{$N=2000,T=60$} & 4.00  & 4.00  & 4.00  & 4.00  & 4.00  & 4.00  \\
		\multicolumn{1}{l}{$N=100,T=100$} & 4.00  & 4.00  & 5.11(0.60) & 4.00  & 4.00  & 4.00  \\
		\multicolumn{1}{l}{$N=40,T=100$} & 4.00(0.03) & 4.00  & 4.69(0.59) & 4.00  & 4.00  & 4.00  \\
		\multicolumn{1}{l}{$N=60,T=100$} & 4.00  & 4.00  & 4.49(0.54) & 4.00  & 4.00  & 4.00  \\
		\multicolumn{1}{l}{$N=60,T=200$} & 4.00  & 4.00  & 4.00  & 4.00  & 4.00  & 4.00  \\
		\multicolumn{1}{l}{$N=10,T=50$} & 8.00  & 8.00  & 8.00  & 4.04(0.47) & 3.97(0.43) & 4.11(0.53) \\
		\multicolumn{1}{l}{$N=10,T=100$} & 8.00  & 8.00  & 8.00  & 3.93(0.25) & 3.93(0.26) & 3.94(0.24) \\
		\multicolumn{1}{l}{$N=20,T=100$} & 5.93(0.67) & 5.35(0.65) & 7.13(0.65) & 4.00  & 4.00  & 4.00  \\
		\multicolumn{1}{l}{$N=100,T=20$} & 5.92(0.66) & 5.35(0.67) & 7.11(0.62) & 4.73(0.66) & 4.38(0.55) & 5.82(0.80) \\
		\hline
	\end{tabular} 
	}
\end{table}

\begin{table}[!t]
	\centering\caption{Comparison between  $PC_{pj}$ and $PC^{*}_{pj}$ for Model 4 }
	\label{tab:model4}
   \scalebox{0.9}{
	\begin{tabular}{lllllll}
		\hline
		\multicolumn{1}{c}{($N, T$)} &  $PC_{p1}$   & $PC_{p1}$ & $PC_{p1}$   & $PC^{*}_{p1}$& $PC^{*}_{p2}$ & $PC^{*}_{p3}$\\
		\hline
		\multicolumn{7}{c}{Mode4 under Gaussian assumption } \\
		\multicolumn{1}{l}{$N=100,T=40$} & 4.00(0.03) & 4.00  & 4.82(0.58) & 4.00  & 4.00  & 4.01(0.09) \\
		\multicolumn{1}{l}{$N=100,T=60$} & 4.00  & 4.00  & 4.81(0.59) & 4.00  & 4.00  & 4.00  \\
		\multicolumn{1}{l}{$N=200,T=60$} & 4.00  & 4.00  & 4.00  & 4.00  & 4.00  & 4.00  \\
		\multicolumn{1}{l}{$N=500,T=60$} & 4.00  & 4.00  & 4.00  & 4.00  & 4.00  & 4.00  \\
		\multicolumn{1}{l}{$N=1000,T=60$} & 4.00  & 4.00  & 4.00  & 4.00  & 4.00  & 4.00  \\
		\multicolumn{1}{l}{$N=2000,T=60$} & 4.00  & 4.00  & 4.00  & 4.00  & 4.00  & 4.00  \\
		\multicolumn{1}{l}{$N=100,T=100$} & 4.00  & 4.00  & 6.08(0.65) & 4.00  & 4.00  & 4.00  \\
		\multicolumn{1}{l}{$N=40,T=100$} & 4.00  & 4.00  & 5.54(0.66) & 4.00  & 4.00  & 4.00  \\
		\multicolumn{1}{l}{$N=60,T=100$} & 4.00  & 4.00  & 5.25(0.62) & 4.00  & 4.00  & 4.00  \\
		\multicolumn{1}{l}{$N=60,T=200$} & 4.00  & 4.00  & 4.00  & 4.00  & 4.00  & 4.00  \\
		\multicolumn{1}{l}{$N=10,T=50$} & 8.00  & 8.00  & 8.00  & 4.067(0.69) & 3.96(0.62) & 4.26(0.82) \\
		\multicolumn{1}{l}{$N=10,T=100$} & 8.00   & 8.00   & 8.00  & 3.90(0.42) & 3.87(0.41) & 3.94(0.44) \\
		\multicolumn{1}{l}{$N=20,T=100$} & 6.98(0.63) & 6.41(0.66) & 7.86(0.38) & 4     & 4.00(0.04) & 4 \\
		\multicolumn{1}{l}{$N=100,T=20$} & 5.85(0.67) & 5.26(0.62) & 7.08(0.64) & 4.67(0.63) & 4.31(0.49) & 5.73(0.80) \\
		\multicolumn{7}{c}{Model4 under Gamma assumption} \\
		\multicolumn{1}{l}{$N=100,T=40$} & 4.02(0.13) & 4.00(0.03) & 5.22(0.64) & 4.00  & 4.00  & 4.08(0.27) \\
		\multicolumn{1}{l}{$N=100,T=60$} & 4.00  & 4.00  & 5.19(0.61) & 4.00  & 4.00  & 4.02(0.13) \\
		\multicolumn{1}{l}{$N=200,T=60$} & 4.00  & 4.00  & 4.00  & 4.00  & 4.00  & 4.00  \\
		\multicolumn{1}{l}{$N=500,T=60$} & 4.00  & 4.00  & 4.00  & 4.00  & 4.00  & 4.00  \\
		\multicolumn{1}{l}{$N=1000,T=60$} & 4.00  & 4.00  & 4.00  & 4.00  & 4.00  & 4.00  \\
		\multicolumn{1}{l}{$N=2000,T=60$} & 4.00  & 4.00  & 4.00  & 4.00  & 4.00  & 4.00  \\
		\multicolumn{1}{l}{$N=100,T=100$} & 4.00  & 4.00  & 6.52(0.68) & 4.00  & 4.00  & 4.01(0.09) \\
		\multicolumn{1}{l}{$N=40,T=100$} & 4.12(0.32) & 4.01(0.11) & 5.91(0.67) & 4.00  & 4.00  & 4.00(0.03) \\
		\multicolumn{1}{l}{$N=60,T=100$} & 4.00(0.04) & 4.00  & 5.64(0.66) & 4.00  & 4.00  & 4.00(0.04) \\
		\multicolumn{1}{l}{$N=60,T=200$} & 4.00  & 4.00  & 4.00  & 4.00  & 4.00  & 4.00  \\
		\multicolumn{1}{l}{$N=10,T=50$} & 8.00  & 8.00  & 8.00  & 4.30(0.83) & 4.18(0.76) & 4.50(0.92) \\
		\multicolumn{1}{l}{$N=10,T=100$} & 8.00  & 8.00  & 8.00  & 3.98(0.56) & 3.94(0.53) & 4.04(0.62) \\
		\multicolumn{1}{l}{$N=20,T=100$} & 7.08(0.64) & 6.56(0.69) & 7.91(0.29) & 4.00  & 4.00(0.03) & 4.01(0.10) \\
		\multicolumn{1}{l}{$N=100,T=20$} & 6.23(0.69) & 5.68(0.66) & 7.36(0.61) & 5.12(0.75) & 4.66(0.63) & 6.19(0.80) \\
		\hline
	\end{tabular}
	}
\end{table}

\begin{table}[!t]
	\centering\caption{Comparison between  $PC_{pj}$ and $PC^{*}_{pj}$ for Model 5 }
	\label{tab:model5}
	\scalebox{0.9}{
	\begin{tabular}{lllllll}
		\hline
		\multicolumn{1}{c}{($N, T$)} &  $PC_{p1}$   & $PC_{p1}$ & $PC_{p1}$   & $PC^{*}_{p1}$& $PC^{*}_{p2}$ & $PC^{*}_{p3}$\\
		\hline
		\multicolumn{7}{c}{Model5 under Gaussian assumption} \\
		\multicolumn{1}{l}{$N=100,T=40$} & 4.00  & 4.00  & 4.85(0.63) & 4.00  & 4.00  & 4.01(0.11) \\
		\multicolumn{1}{l}{$N=100,T=60$} & 4.00  & 4.00  & 4.78(0.58) & 4.00  & 4.00  & 4.00  \\
		\multicolumn{1}{l}{$N=200,T=60$} & 4.00  & 4.00  & 4.00  & 4.00  & 4.00  & 4.00  \\
		\multicolumn{1}{l}{$N=500,T=60$} & 4.00  & 4.00  & 4.00  & 4.00  & 4.00  & 4.00  \\
		\multicolumn{1}{l}{$N=1000,T=60$} & 4.00  & 4.00  & 4.00  & 4.00  & 4.00  & 4.00  \\
		\multicolumn{1}{l}{$N=2000,T=60$} & 4.00  & 4.00  & 4.00  & 4.00  & 4.00  & 4.00  \\
		\multicolumn{1}{l}{$N=100,T=100$} & 4.00  & 4.00  & 6.09(0.65) & 4.00  & 4.00  & 4.00  \\
		\multicolumn{1}{l}{$N=40,T=100$} & 4.02(0.13) & 4.00  & 5.52(0.64) & 4.00  & 4.00  & 4.00  \\
		\multicolumn{1}{l}{$N=60,T=100$} & 4.00  & 4.00  & 5.26(0.64) & 4.00  & 4.00  & 4.00  \\
		\multicolumn{1}{l}{$N=60,T=200$} & 4.00  & 4.00  & 4.00  & 4.00  & 4.00  & 4.00  \\
		\multicolumn{1}{l}{$N=10,T=50$} & 8.00  & 8.00  & 8.00  & 4.11(0.73) & 4.00(0.65) & 4.29(0.85) \\
		\multicolumn{1}{l}{$N=10,T=100$} & 8.00  & 8.00  & 8.00  & 3.96(0.40) & 3.90(0.39) & 3.97(0.44) \\
		\multicolumn{1}{l}{$N=20,T=100$} & 6.99(0.63) & 6.39(0.63) & 7.87(0.34) & 4.00(0.03) & 4.00(0.04) & 4.00(0.03) \\
		\multicolumn{1}{l}{$N=100,T=20$} & 5.87(0.66) & 5.31(0.65) & 7.13(0.64) & 4.68(0.64) & 4.29(0.47) & 5.76(0.79) \\
		\multicolumn{7}{c}{Model5 under Gamma assumption} \\
		\multicolumn{1}{l}{$N=100,T=40$} & 4.01(0.12) & 4.00  & 5.19(0.65) & 4.00  & 4.00  & 4.06(0.24) \\
		\multicolumn{1}{l}{$N=100,T=60$} & 4.00  & 4.00  & 5.15(0.64) & 4.00  & 4.00  & 4.00(0.08) \\
		\multicolumn{1}{l}{$N=200,T=60$} & 4.00  & 4.00  & 4.00  & 4.00  & 4.00  & 4.00  \\
		\multicolumn{1}{l}{$N=500,T=60$} & 4.00  & 4.00  & 4.00  & 4.00  & 4.00  & 4.00  \\
		\multicolumn{1}{l}{$N=1000,T=60$} & 4.00  & 4.00  & 4.00  & 4.00  & 4.00  & 4.00  \\
		\multicolumn{1}{l}{$N=2000,T=60$} & 4.00  & 4.00  & 4.00  & 4.00  & 4.00  & 4.00  \\
		\multicolumn{1}{l}{$N=100,T=100$} & 4.00  & 4.00  & 6.41(0.67) & 4.00  & 4.00  & 4.00(0.05) \\
		\multicolumn{1}{l}{$N=40,T=100$} & 4.08(0.27) & 4.00(0.06) & 5.86(0.67) & 4.00  & 4.00  & 4.00  \\
		\multicolumn{1}{l}{$N=60,T=100$} & 4.00  & 4.00  & 5.58(0.64) & 4.00  & 4.00  & 4.00  \\
		\multicolumn{1}{l}{$N=60,T=200$} & 4.00  & 4.00  & 4.00  & 4.00  & 4.00  & 4.00  \\
		\multicolumn{1}{l}{$N=10,T=50$} & 8.00  & 8.00  & 8.00  & 4.27(0.79) & 4.11(0.71) & 4.50(0.90) \\
		\multicolumn{1}{l}{$N=10,T=100$} & 8.00  & 8.00  & 8.00  & 3.98(0.50) & 3.95(0.49) & 4.05(0.57) \\
		\multicolumn{1}{l}{$N=20,T=100$} & 7.15(0.62) & 6.62(0.64) & 7.92(0.28) & 4.00  & 4.00  & 4.0090.06) \\
		\multicolumn{1}{l}{$N=100,T=20$} & 6.23(0.68) & 5.67(0.65) & 7.36(0.60) & 5.07(0.71) & 4.63(0.63) & 6.20(0.79) \\
		\hline
	\end{tabular} 
	}
\end{table}

Reported in Tables \ref{tab:model3} to \ref{tab:model5} are the empirical means of the estimations of the number of factors over 1,000 replications, corresponding to Model 3 to 5 respectively. The standard errors are also given in parentheses following the estimations. If the standard error is 0, no further annotations will be made. Refer to \cite{Pass2017}, for these six scenarios, the predetermined maximum number $m_0$ of factors is set to 8. 

As shown in the simulated results, our proposed information criteria $P C^{*}_{p }$  have an overall better performance
than $PC_{p}$ in \cite{BN2002} for different models and populations. When $\text{min}(N,T) > 40$, both information criteria $PC_{P}$ and $PC^*_{P}$ can obtain satisfactory  estimation of the number of  factors. But our estimation are more accurate  with smaller standard errors.
When $\text{min}(N,T) < 40$, the original information criteria $PC_{P}$ almost   report the predetermined maximum value  8, but the estimation of our method is closer to the true value 4. In addition, when the population assumption is gamma distribution, the efficiency of detecting the true number of factors of $PC_{p}$ in \cite{BN2002} is lower than that of Gaussian distribution, but  our new information criteria $P C^{*}_{p }$ perform still well for the non-Gaussian assumptions. When the error  sequence no longer satisfy the independent assumption, our method also outperforms  $PC_{p}$ in \cite{BN2002}.


\section{Real data analysis}\label{isect5}
The proposed information criteria seem to work rather well 
in the simulated experiments. We now apply the new procedure of determining the number of the factors to two real data sets. Both data sets are downloaded from  \href{https://www.oecd.org/}{\textit{https://www.oecd.org}}. The first  is the OECD Composite Leading Indicators (CLI) data set, which is constructed by weighting indicator data in various fields of the national economy according to certain standards. It is a leading indicator reflecting a country's macroeconomic development cycle. And this data set consists of CLI for 32 OECD countries, 6 non-member economies and 8 zone aggregates ($N=46$) observed monthly from June 1998 to October 2020 ($T=269$).  The second data set contains OECD Business Confidence Indicators (BCI) data for 36 OECD countries, 6 non-member economies and 6 zone aggregates ($N=48$) observed monthly from November 2003 to May 2021 ($T=211$). 

In practice analysis, we need  to pay attention to the chosen of $m_0$, which will effect the robustness of estimations.
Since the $N  \times N$ matrix $ \mathbf \Lambda \mathbf \Lambda'$ has only $M$ non-zero spiked eigenvalues, but the rest tailed eigenvalues are all zero, so that we only need to select an arbitrary large integer $m_0$, and then we can obtain the robust estimations. But in practical applications,  except the $M$  spiked eigenvalues,  the tailed eigenvalues of $\mathbf \Lambda \mathbf \Lambda'$ are not all zeros, but most of them are relatively small values close to zero.  
Due to this reason, $m_0$ should be selected to sufficiently ensure the inclusion of the most of non-zero eigenvalues. According to practical experience, we adopt the selection of $m_0$  in the range of $0.6N$ to $0.8N$.

 
 Then, we apply the information criteria $PC_{p}$  of  \cite{BN2002}  and our proposed criteria $PC^{*}_{p}$ to estimate  the number of factors for the real data sets. The estimate results  on the two data sets are shown in Table \ref{data1} and  \ref{data2}. 
 The three rows of the tables correspond to the estimations when  $m_0$ is selected as  $0.6N$, $0.7N$ and $0.8N$, respectively.
 It suggests that the original criteria $PC_{pj}$'s seriously overestimate for the two data sets. In contrast, when $m_0=0.7N$, our proposed information criteria $PC^*_{pj}$'s estimate the number of factors for both data sets to be 10 or 11, and the corresponding cumulative contribution rate can reach 95\%.  It implies that 10 or 11 is a reasonable estimation of the number of factors for both data sets.
\begin{table}[！t]
	\centering
	\caption{Comparison between  $PC_{pj}$ and $PC^{*}_{pj}$ on the first data set}
	\begin{tabular}{ccccccc}
		\hline
		\multicolumn{1}{c}{} & $PC_{p1}$  & $PC_{p2}$ & $PC_{p3}$  & $PC^{*}_{p1}$& $PC^{*}_{p2}$ & $PC^{*}_{p3}$\\
		\hline
		\multicolumn{1}{c}{$m_{0}=28$} &24 & 23 & 27 & 8  & 8 & 9 \\
		\multicolumn{1}{c}{$m_{0}=32$} &32 & 32 & 32 & 10  &10 & 11\\
		\multicolumn{1}{c}{$m_{0}=37$} &37 & 37 & 37 & 13  & 13 & 14\\
		\hline
	\end{tabular}%
	\label{data1}%
\end{table}%

\begin{table}[！t]
	\centering
	\caption{Comparison between  $PC_{pj}$ and $PC^{*}_{pj}$ on the second data set }
	\begin{tabular}{ccccccc}
		\hline
		\multicolumn{1}{c}{} & $PC_{p1}$  & $PC_{p2}$ & $PC_{p3}$  & $PC^{*}_{p1}$& $PC^{*}_{p2}$ & $PC^{*}_{p3}$\\
		\hline
		\multicolumn{1}{c}{$m_{0}=29$} &23 & 22 & 27 & 8  & 8 & 9 \\
		\multicolumn{1}{c}{$m_{0}=34$} &34 & 34 & 34 & 10  &10 & 11\\
		\multicolumn{1}{c}{$m_{0}=38$} &38 & 38 & 38 & 12  & 12 & 13\\
		\hline
	\end{tabular}%
	\label{data2}%
\end{table}%

\section{Conclusion}
This paper aimed to determine the number of factors in a large-dimensional generalised factor model with more relaxed assumptions than that of previous works.
 For the target model, we introduced the bias-corrected noise estimator $\hat \sigma^{2}_*$ by random matrix theory,  further construct the  information criteria $PC^{*}_{p}$  based on $\hat \sigma^{2}_*$, and  estimate the number of factors consistently. The good performance of our method is demonstrated by simulations and empirical applications.
 This paper only focused on the static factor models. Further we will improve the information criteria to accommodate the dynamic factor models in the future work.

\appendix
\section{The logarithm transformed MAEs among several noise estimators} \label{app:noise1}

\begin{table}[htbp]
  \centering
  \caption{The logarithm transformed MAEs among several noise estimators in Model 1}
    \begin{tabular}{lllllll}
    \hline
	    Estimators & \multicolumn{1}{c}{$\hat{\sigma}_{*}^{2}$} & \multicolumn{1}{c}{ $\hat{\sigma}_{MLE}^{2}$} & \multicolumn{1}{c}{$\hat{\sigma}_{\mathrm{P}}^2$} & \multicolumn{1}{c}{$\hat{\sigma}_{\mathrm{KN}}^{2}$} & \multicolumn{1}{c}{$\hat{\sigma}_{\mathrm{US}}^{2}$} & \multicolumn{1}{c}{$\hat{\sigma}_{\text {median }}^{2}$} \\
	 \hline
		\multicolumn{7}{c}{ Under Gaussian assumption and $c=0.5$} \\
    $N=50,T=100$ & -0.95594 & -0.21996 & 1.079281 & 1.061465 & 0.448425 & 1.023963 \\
    $N=100,T=200$ & -1.29808 & -0.51528 & 1.113894 & 1.105469 & 0.535754 & 1.04767 \\
    $N=150,T=300$ & -1.48663 & -0.68761 & 1.124752 & 1.119236 & 0.565924 & 1.059354 \\
    $N=200,T=400$ & -1.60601 & -0.81515 & 1.130371 & 1.129887 & 0.582757 & 1.067715 \\
    $N=250,T=500$ & -1.70851 & -0.90724 & 1.133389 & 1.133012 & 0.59349 & 1.074012 \\
    $N=300,T=600$ & -1.78126 & -0.98528 & 1.13551 & 1.135202 & 0.600825 & 1.078858 \\
    $N=350,T=700$ & -1.86527 & -1.05624 & 1.137166 & 1.136905 & 0.605671 & 1.082972 \\
    $N=400,T=800$ & -1.92219 & -1.11401 & 1.138295 & 1.138069 & 0.609859 & 1.086277 \\
         \multicolumn{7}{c}{Under Gaussian assumption and $c=1.5$} \\
    $N=90,T=60$ & -0.80542 & -0.2007 & 1.089998 & 1.062031 & 0.571411 & 0.967804 \\
    $N=150,T=100$ & -0.98767 & -0.40775 & 1.112526 & 1.095778 & 0.638532 & 0.992417 \\
    $N=210,T=140$ & -1.10124 & -0.55065 & 1.122464 & 1.110508 & 0.664882 & 1.010448 \\
    $N=270,T=180$ & -1.18679 & -0.66006 & 1.128019 & 1.11872 & 0.679631 & 1.024954 \\
    $N=330,T=220$ & -1.27469 & -0.74316 & 1.131205 & 1.123598 & 0.690476 & 1.030937 \\
    $N=390,T=260$ & -1.35075 & -0.81195 & 1.13339 & 1.12695 & 0.695654 & 1.039021 \\
    $N=450,T=300$ & -1.41206 & -0.87197 & 1.135046 & 1.129465 & 0.700718 & 1.043806 \\
    $N=510,T=340$ & -1.44716 & -0.93001 & 1.136531 & 1.131607 & 0.704886 & 1.050749 \\
          \multicolumn{7}{c}{Under Gamma assumption and $c=0.5$} \\
    $N=50,T=100$ & -0.7978 & -0.20039 & 1.074523 & 1.056506 & 0.465334 & 0.972263 \\
    $N=100,T=200$ & -1.13773 & -0.50752 & 1.113048 & 1.104576 & 0.547863 & 1.003191 \\
    $N=150,T=300$ & -1.31941 & -0.68424 & 1.124514 & 1.118972 & 0.57605 & 1.017331 \\
    $N=200,T=400$ & -1.45258 & -0.80535 & 1.129859 & 1.129375 & 0.589596 & 1.031058 \\
    $N=250,T=500$ & -1.548 & -0.9029 & 1.133208 & 1.132832 & 0.597788 & 1.038072 \\
    $N=300,T=600$ & -1.64719 & -0.97789 & 1.135253 & 1.134946 & 0.605367 & 1.044539 \\
    $N=350,T=700$ & -1.70155 & -1.05285 & 1.137067 & 1.136807 & 0.61079 & 1.050343 \\
    $N=400,T=800$ & -1.74588 & -1.11015 & 1.138197 & 1.137971 & 0.614225 & 1.055111 \\
          \multicolumn{7}{c}{Under Gamma assumption and $c=1.5$} \\
    $N=50,T=100$ & -0.81382 & -0.17868 & 1.084367 & 1.056048 & 0.5545 & 0.902996 \\
    $N=100,T=200$ & -0.98342 & -0.3935 & 1.110467 & 1.093598 & 0.628851 & 0.935026 \\
    $N=150,T=300$ & -1.09768 & -0.53948 & 1.121347 & 1.109331 & 0.658228 & 0.955865 \\
    $N=200,T=400$ & -1.18749 & -0.64851 & 1.12714 & 1.117806 & 0.675449 & 0.973028 \\
    $N=250,T=500$ & -1.262 & -0.73747 & 1.130854 & 1.123223 & 0.685124 & 0.982414 \\
    $N=300,T=600$ & -1.32379 & -0.80877 & 1.133225 & 1.126769 & 0.692978 & 0.99098 \\
    $N=350,T=700$ & -1.38851 & -0.86749 & 1.134844 & 1.12925 & 0.697477 & 1.000551 \\
    $N=400,T=800$ & -1.43185 & -0.92462 & 1.136319 & 1.131384 & 0.701558 & 1.00541 \\
    \hline
    \end{tabular}%
\end{table}%

\begin{table}[htbp]
  \centering
  \caption{The logarithm transformed MAEs among several noise estimators in Model 2}
  \label{app:noise2}
    \begin{tabular}{lllllll}
    \hline
	    Estimators & \multicolumn{1}{c}{$\hat{\sigma}_{*}^{2}$} & \multicolumn{1}{c}{ $\hat{\sigma}_{MLE}^{2}$} & \multicolumn{1}{c}{$\hat{\sigma}_{\mathrm{P}}^2$} & \multicolumn{1}{c}{$\hat{\sigma}_{\mathrm{KN}}^{2}$} & \multicolumn{1}{c}{$\hat{\sigma}_{\mathrm{US}}^{2}$} & \multicolumn{1}{c}{$\hat{\sigma}_{\text {median }}^{2}$} \\
	 \hline
		\multicolumn{7}{c}{Under Gaussian assumption and $c=0.5$} \\
    $N=50,T=100$ & -0.95278 & -0.2181 & 1.07884 & 1.061015 & 0.448869 & 1.208373 \\
    $N=100,T=200$ & -1.30545 & -0.51772 & 1.114157 & 1.105739 & 0.535016 & 1.180092 \\
    $N=150,T=300$ & -1.4831 & -0.68824 & 1.124797 & 1.119282 & 0.566147 & 1.169562 \\
    $N=200,T=400$ & -1.60578 & -0.81338 & 1.130279 & 1.129795 & 0.582211 & 1.16414 \\
    $N=250,T=500$ & -1.70975 & -0.91001 & 1.133503 & 1.133127 & 0.592704 & 1.160719 \\
    $N=300,T=600$ & -1.77951 & -0.98689 & 1.135565 & 1.135257 & 0.600081 & 1.158251 \\
    $N=350,T=700$ & -1.84042 & -1.0515 & 1.137028 & 1.136767 & 0.605395 & 1.156618 \\
    $N=400,T=800$ & -1.926 & -1.11542 & 1.13833 & 1.138105 & 0.609938 & 1.155565 \\
          \multicolumn{7}{c}{Under Gaussian assumption and $c=1.5$} \\
    $N=90,T=60$ & -0.81416 & -0.19829 & 1.0894 & 1.061424 & 0.572178 & 1.178899 \\
    $N=150,T=100$ & -0.971 & -0.41135 & 1.113034 & 1.096288 & 0.640573 & 1.16716 \\
    $N=210,T=140$ & -1.09756 & -0.5519 & 1.122587 & 1.11063 & 0.666064 & 1.161364 \\
    $N=270,T=180$ & -1.20135 & -0.65616 & 1.127725 & 1.118427 & 0.680227 & 1.157876 \\
    $N=330,T=220$ & -1.2806 & -0.74153 & 1.131106 & 1.123497 & 0.689871 & 1.155746 \\
    $N=390,T=260$ & -1.33833 & -0.81492 & 1.133544 & 1.127105 & 0.696498 & 1.154434 \\
    $N=450,T=300$ & -1.40564 & -0.874 & 1.135137 & 1.129555 & 0.700147 & 1.1532 \\
    $N=510,T=340$ & -1.45006 & -0.9301 & 1.136534 & 1.13161 & 0.703742 & 1.152568 \\
          \multicolumn{7}{c}{Under Gamma assumption and $c=0.5$} \\
    $N=50,T=100$ & -0.85259 & -0.20869 & 1.076574 & 1.058641 & 0.442791 & 1.204876 \\
    $N=100,T=200$ & -1.17927 & -0.50421 & 1.112681 & 1.104226 & 0.533345 & 1.176478 \\
    $N=150,T=300$ & -1.34406 & -0.68147 & 1.124317 & 1.118785 & 0.565031 & 1.167195 \\
    $N=200,T=400$ & -1.49708 & -0.80743 & 1.129969 & 1.129485 & 0.582381 & 1.162066 \\
    $N=250,T=500$ & -1.58779 & -0.90644 & 1.133356 & 1.132979 & 0.594179 & 1.159008 \\
    $N=300,T=600$ & -1.67897 & -0.98786 & 1.135598 & 1.13529 & 0.601305 & 1.156902 \\
    $N=350,T=700$ & -1.74402 & -1.04957 & 1.136971 & 1.136711 & 0.606059 & 1.155449 \\
    $N=400,T=800$ & -1.81112 & -1.10866 & 1.138159 & 1.137933 & 0.610654 & 1.154233 \\
          \multicolumn{7}{c}{Under Gamma assumption and $c=1.5$} \\
    $N=50,T=100$ & -0.81058 & -0.18644 & 1.086393 & 1.058178 & 0.548338 & 1.168817 \\
    $N=100,T=200$ & -0.96384 & -0.40608 & 1.11229 & 1.095475 & 0.625941 & 1.159911 \\
    $N=150,T=300$ & -1.08621 & -0.54748 & 1.12215 & 1.110156 & 0.656867 & 1.156134 \\
    $N=200,T=400$ & -1.17893 & -0.65703 & 1.127791 & 1.118473 & 0.675113 & 1.154124 \\
    $N=250,T=500$ & -1.27892 & -0.73578 & 1.130749 & 1.123124 & 0.683333 & 1.15202 \\
    $N=300,T=600$ & -1.34049 & -0.80886 & 1.133229 & 1.126779 & 0.691245 & 1.151347 \\
    $N=350,T=700$ & -1.38994 & -0.87447 & 1.135158 & 1.129571 & 0.696488 & 1.150754 \\
    $N=400,T=800$ & -1.44307 & -0.9276 & 1.136436 & 1.131505 & 0.700241 & 1.150307 \\
    \hline
    \end{tabular}%
\end{table}%

\nocite{*}

\bibliographystyle{ejbib}

\bibliography{arXiv-Determine_factor_number}

\end{document}